\documentclass[11pt]{article}

\usepackage{graphicx}
\usepackage{latexsym}
\usepackage{amssymb}
\usepackage{amsmath}
\usepackage{amsthm}
\usepackage{bm}

\usepackage{mathtools}

\usepackage{nth}

\usepackage{forloop}
\usepackage[paper=letterpaper,margin=1in]{geometry}
\usepackage{paralist}
\usepackage{nicefrac}
\usepackage{rotating}
\usepackage{array,multirow}
\usepackage{threeparttable}
\usepackage{hyperref}

\newtheorem{thm}{Theorem}

\newtheorem{lemma}[thm]{Lemma}

\newtheorem{obs}[thm]{Observation}

\newcommand\E{\mathbb{E}}
\newcommand\R{\mathbb{R}}

\newcommand\disc{\mathrm{disc}}

\date{\vspace{-5ex}}
\begin{document}

\title{An Algorithm for Koml\'{o}s Conjecture Matching Banaszczyk's Bound}

\author{Nikhil Bansal
\thanks{Department of Mathematics and Computer Science, Eindhoven University of Technology, Netherlands.  
Email:
\href{mailto:n.bansal@tue.nl}{n.bansal@tue.nl}.
Supported by a NWO Vidi grant 639.022.211 and an ERC consolidator grant 617951.}
\and
Daniel Dadush\thanks{Centrum Wiskunde \& Informatica, Amsterdam.
\texttt{dadush@cwi.nl}. Supported by the NWO Veni grant 639.071.510.}
\and
Shashwat Garg\thanks{Department of Mathematics and Computer Science, Eindhoven University of Technology, Netherlands.  
Email:
\href{mailto:s.garg@tue.nl}{s.garg@tue.nl}.
Supported by the Netherlands' Organisation for Scientific Research (NWO) under project no.~022.005.025.
}
}

\maketitle

\begin{abstract}
We consider the problem of finding a low discrepancy coloring for sparse set systems where each element lies in at most $t$ sets. We give an efficient algorithm that finds a coloring with discrepancy 
$O((t \log n)^{1/2})$, matching the best known non-constructive bound for the problem due to Banaszczyk.
The previous algorithms only achieved an $O(t^{1/2} \log n)$ bound.
The result also extends to the more general Koml\'{o}s setting and gives an algorithmic
$O(\log^{1/2} n)$ bound.
\end{abstract}

\section{Introduction}

Let $(V,\mathcal{S})$ be a finite set system, with $V=\{1,\ldots,n\}$ and $\mathcal{S} = \{S_1,\ldots,S_m\}$ a collection of subsets of $V$. For a two-coloring $\chi: V  \rightarrow \{-1,1\}$, the discrepancy of $\chi$ for a set $S$  is defined as $ \chi(S) =  |\sum_{i\in S} \chi(i) |$ and measures the imbalance from an even-split for $S$.
The discrepancy of the system $(V,\mathcal{S})$ is defined as 
\[ \disc(\mathcal{S}) = \min_{\chi:V \rightarrow \{-1,1\}} \max_{S \in \mathcal{S}} \chi(S). \]
That is, it is the minimum imbalance for all sets in $\mathcal{S}$, over all possible two-colorings $\chi$.

Discrepancy is a widely studied topic and has applications to many areas in mathematics and computer science. For more background we refer the reader to the books \cite{Chazelle,Mat99,Panorama}. In particular, discrepancy is closely related to the problem of rounding  fractional solutions of a linear system of equations to integral ones \cite{LSV,R12}, and is widely studied in approximation algorithms and optimization.

Until recently, most of the results in discrepancy were based on non-algorithmic approaches 
and hence were not directly useful for algorithmic applications. 
However, in the last few years there has been remarkable progress in our understanding of the algorithmic aspects of discrepancy \cite{B10,CNN11,LM12,Ro14,HSS14,ES14,NT15}. In particular, we can now match or even improve upon all known applications of the widely used partial-coloring method \cite{Spencer85,Mat99} in discrepancy.
This has, for example, led to several other new results in approximation algorithms \cite{R13,BCKL14,BN15,NTZ13}.


\paragraph*{Sparse Set Systems} 
Despite the algorithmic progress, one prominent question that has remained open is to match the known non-constructive bounds on discrepancy for low degree or sparse set systems. These systems are parametrized by $t$, that denotes the maximum number of sets that contain any element.
 Beck and Fiala \cite{BF81} proved, using an algorithmic iterated rounding approach, that any such set system has discrepancy at most $2t-1$. They also conjectured that the discrepancy in this case is $O(t^{1/2})$, and settling this has been an elusive open problem.

The best known result in this direction is due to Banaszczyk \cite{B97}, which implies an $O(\sqrt{t \log n})$ discrepancy bound for the problem\footnote{We assume here that $t \geq \log n$, otherwise the $O(t)$ bound is better.}. Unlike most results in discrepancy that are based on the partial-coloring method,
Banaszczyk's proof is based on a very different and elegant convex geometric argument, and it is not at all clear how to make it algorithmic.
Prior to Banaszczyk's result, the best known non-algorithmic bound was $O(t^{1/2} \log n)$ \cite{Srin97}, based on the partial-coloring method. This bound was first made algorithmic in \cite{B10}, and by now there are several different ways known to obtain this result \cite{LM12,N13,Ro14,HSS14,ES14}.
However the question of matching Banaszczyk's bound algorithmically for the problem and its variants has been open despite a lot of attention in recent years \cite{N13,B13,ES14,EL15}. In particular, as we discuss in Section \ref{s:overview} there is a natural algorithmic barrier to improving the $O(t^{1/2} 
\log n)$ bound. 

A substantial generalization of the Beck-Fiala conjecture is the following:

\smallskip

 \noindent
\textit{Koml\'{o}s Conjecture:} Given any collection of vectors $v_1,\ldots,v_n \in \mathbb{R}^m$ such that $\|v_i\|_2 \leq 1$ for each $i\in[n]$\footnote{We use $[n]$ to denote the set $\{1,2,\dots,n\}$.}, there exist signs $x_1,\ldots, x_n \in \{-1,1\}$ such that $\|\sum_{i=1}^n x_i v_i \|_\infty =O(1)$.

\smallskip

This implies the Beck-Fiala conjecture by choosing each $v_i$ as the column corresponding to element $i$ in the incidence matrix of the set system scaled  by $t^{-1/2}$.
Again, the best known non-constructive bound here is $O(\sqrt{\log n})$ due to  
Banaszczyk and  the previous algorithmic techniques can also be adapted to achieve $O(\log n)$ constructively for the Koml\'{o}s setting. 

\subsection{Our Results}
In this paper we give the following algorithmic result 
 for the Beck-Fiala problem, which matches the non-constructive bound due to Banaszczyk.

\begin{thm}
\label{thm:main}
Given a set system $(V,\mathcal{S})$ with $|V|=n$ such that each element $i\in V$ lies in at most $t$ sets in $\mathcal{S}$, there is an efficient randomized algorithm that finds an $O(\sqrt{t\log n})$ discrepancy coloring with high probability.
\end{thm}

Our result also extends to the Koml\'{o}s setting with some minor modifications.

\begin{thm}
\label{komlos}
Given an $m\times n$ matrix $A$ with all columns of $\ell_2$-norm at most $1$, there is an efficient randomized algorithm that finds $x\in\{-1,1\}^n$ such that $\| Ax\|_\infty=O(\sqrt{\log n})$ with high probability.
\end{thm}

Our algorithm gives a new constructive proof of Banaszczyk's result for the Beck-Fiala and Koml\'{o}s setting.
While Theorem \ref{komlos} implies Theorem \ref{thm:main}, for better clarity we first present the algorithm for the Beck-Fiala problem in Sections \ref{s:bf} and \ref{s:analysis} and then discuss the extension to Theorem \ref{komlos} in Section \ref{s:komlos}.

\subsection{High-level Overview}
\label{s:overview}
The algorithm has a similar structure to the previous random walk based approaches \cite{B10,LM12,HSS14}.
It starts with the coloring $x_0=0^n$ at time $0$, and at each time step $k$, updates the color of element $i$
by adding a small increment to its coloring at time $k-1$, i.e.~$x_k(i)=x_{k-1}(i) + \Delta x_k(i)$. 
If a variable reaches $-1$ or $1$ it is frozen, and its value is not updated any more. 
The increment is determined by solving an appropriate SDP and projecting the resulting vectors in a random direction.


However, all the previous approaches get stuck at the $O(t^{1/2} \log n)$ barrier, and it is instructive to understand why this happens before we present our algorithm.

\paragraph*{The $O(t^{1/2} \log n)$ barrier}
 
Roughly speaking, the execution of the previous algorithms can be divided into $O(\log n)$ phases (either implicitly or explicitly), where in each phase about half the variables get frozen and  each set incurs an expected discrepancy of $O(t^{1/2})$. This gives an overall discrepancy bound of $O(t^{1/2} \log n)$. 

Intuitively however, for a fixed set $S$, one should expect an $O(t^{1/2})$ discrepancy over {\em all} the phases for the following reason.
Assume that all sets are of size $O(t)$. This can be ensured using a standard linear algebraic argument to ensure that sets incur zero discrepancy as long as they have at least $2t$ uncolored elements\footnote{The reader may observe that if all sets were of size $O(t)$, a simple application of the Lov\'{a}sz Local Lemma already gives an $O(\sqrt{t \log t})$ discrepancy coloring, so this should an imply an $O(\sqrt{t \log t})$ discrepancy in general. 
However, the problem is that the Lov\'{a}sz Local Lemma does not combine with the linear algebraic argument.}. 
After $i$ phases of partial coloring, one would expect that $S$ has about $2^{-i}$ fraction of its elements left uncolored, and hence it should incur about $O((2^{-i}t)^{1/2})$ discrepancy in the next phase, giving a total discrepancy of $O(\sum_i 2^{-i/2} t^{1/2}) = O(t^{1/2})$. 
 

However, the problem is that the size of sets may not evolve in this ideal manner, as the partial coloring phase does not give us a fine-grained control over how the elements of each set get colored. For example, even though half of the variables (globally) get colored during a phase, it is possible that half the sets get almost completely colored, while the other half only get $t/\log n$ of their elements colored (while still incurring an $\Omega(t^{1/2})$ discrepancy). 
This imbalance between the discrepancy incurred and ``progress" made for {\em each} set is the fundamental barrier in overcoming the $O(t^{1/2} \log n)$ bound.

\paragraph*{Our approach}
The key idea of our algorithm is to ensure that during the coloring updates the {\em squared} discrepancy we add to a set 
is proportional to the ``progress" elements of that set make towards geting colored. More formally, the updates $\Delta x_k(i)$ that we choose at time $k$ satisfy the following properties:

\begin{enumerate}
\item {\em Zero Discrepancy for large sets:} If a set $S$ has more than $at$ unfrozen (alive) elements at time $k$, for some constant $a$, we ensure that $\sum_{i \in S} \Delta x_k(i) =0$. This is similar to the previous approaches and allows us to  not worry about the discrepancy of a set until its size falls below $at$.

\item {\em Proportional Discrepancy Property:} This is the key new property and (roughly speaking) ensures that the squared discrepancy added to a set is proportional to the ``energy" injected into the set. That is,
\[   \left( \sum_{i \in S} \Delta x_k(i) \right)^2 \leq  2  \left(\sum_{i \in S}  \Delta x_k(i)^2 \right).\]
Note that the left hand side is the square of the  discrepancy increment for set $S$,
and the right hand side is the sum of squares of the increments of the elements of $S$.

Given a coloring $x_k$, let us define the energy of set $S$ at time $k$ as $\sum_{i \in S} x_k(i)^2$. Clearly, the energy of a set can never exceed its size $|S|$. As we can assume that $|S|=O(t)$ (by the Zero Discrepancy Property above), this property  suggests that if the total energy injected ($\sum_k (\sum_{i \in S}  \Delta x_k(i)^2)$) into $S$ was comparable to its final energy (which is $O(t)$) and the increments were mean $0$ random variables,  the squared discrepancy should be $O(t)$.

\item {\em Approximate Orthogonality Constraints} to relate the injected energy to actual energy:
One big problem with the above idea is that the total injected energy into a constraint may be unrelated to its final energy. For example, even 
for a single variable $i$ if the coloring $x_k(i)$ ``fluctuates" a lot around $0$ over time, the injected energy $\sum_k \Delta x_k(i)^2$
could be arbitrarily large, while the final energy for $i$ is at most $1$. For general sets $S$, other problems can arise beyond just fluctuations due to correlations between the updates of different elements of $S$.

To fix this we use the following idea. Suppose we could ensure that for each set $S$ the coloring update at time $k$ was orthogonal 
to the coloring at time $k-1$, i.e.~$ \sum_{i \in S} x_{k-1}(i) \Delta x_k(i) =0$. Then, by Pythagoras theorem, the increase in energy of $S$ would satisfy
\begin{eqnarray}
 \sum_{i \in S} x_k(i)^2 -  \sum_{i \in S} x_{k-1}(i)^2  
 &=  &  \sum_{i \in S} \left((x_{k-1}(i)  + \Delta x_k(i))^2 - x_{k-1}(i)^2 \right) \nonumber \\
& = &  2 \sum_{i \in S} x_{k-1}(i) \Delta x_k(i) + \sum_{i \in S} \Delta x_k(i)^2  \nonumber \\
& = &  \sum_{i \in S} \Delta x_k(i)^2 \label{eq:e-inject}
\end{eqnarray}
where the last equality follows from the orthogonality constraint.
As the expression in \eqref{eq:e-inject} is the injected energy at time $k$, this would precisely make the total injected equal to the final energy as desired.

However, we cannot add such constraints directly for each small set as there might be too many of them. So the 
idea is to add a weaker version of these orthogonality constraints, where we only require that 
 \[ \left(\sum_{i \in S} x_{k-1}(i) \Delta x_k(i)\right)^2 \le 2  \left( \sum_{i \in S} \Delta x_k(i)^2\right) \]
and show that these suffice for our purpose.
\item {\em Sufficient Progress Property}: Of course, all the properties above can be trivially satisfied  by setting $\Delta x_k(i)=0$ for each $i$. So the final step is to ensure that a non-trivial update exists. To this end, we show that 
there exist updates with the (unnormalized) sum $\sum_{i} \Delta x_k(i)^2 = \Omega(A_k)$, where $A_k$ is the number of alive variables at time $k$.

For this purpose, we write an SDP that captures the above constraints and use duality to show the existence of a large feasible solution.
\end{enumerate}

A weaker version of these properties was used in the unpublished manuscript \cite{BG16} to get a more size-sensitive discrepancy bound for each set, but it still only achieved an $O(t^{1/2}\log n)$ discrepancy in the worst case.

We now describe the algorithm and the SDP we use in Section \ref{s:bf}. The analysis consists of two main parts. In Section \ref{s:sdpfeasibility}
we show the sufficient progress property mentioned above, and in Section \ref{s:disc} we show how this gives an overall discrepancy bound of $O((t \log n)^{1/2})$.

\section{Algorithm for the Beck-Fiala Problem}
\label{s:bf}

We will index time by $k$.
Let $x_k\in [-1,1]^n$ denote the coloring at the {\em end} of time step $k$. During the algorithm,
variables which get set to at least $(1-1/n)$ in absolute value are called frozen and their values are not changed anymore. The remaining variables are called alive. We denote by $A_k$ the set of alive variables at the {\em beginning} of time step $k$. Initially all variables are alive. Let $\gamma = 1/(n^2 \log n)$, $T= (12/\gamma^2) \log n$ and $a=6$.

We will call a set $S\in\mathcal{S}$ {\em big} at time $k$ if it has at least $at$ variables alive at time $k$, i.e.~$|S\cap A_k| \geq at$ and {\em small} otherwise.
We will use $\mathcal{B}_k$ to denote the collection of big sets at time $k$ and $\mathcal{L}_k$ to denote the  collection of small (little) sets.

\medskip


\noindent 
{\bf Algorithm:}
\begin{enumerate}
\item Initialize $x_0(i) =0$ for all $i\in [n]$ and $A_1=\{1,2,...,n\}$. 
\item  For each time step $k=1,2,\ldots, T$ repeat the following:
\begin{enumerate}
\item \label{apx:step1}
 Find a solution to the following semidefinite optimization problem:
\begin{eqnarray}
	\textrm{Maximize }\sum_{i\in A_k} \|u_i\|_2^2 \nonumber\\
	\textrm{s.t.} \qquad
	\label{sdp1}  \|\sum_{i \in S\cap A_k} u_i\|_2^2 & = &  \ 0 \qquad  \textrm{for each }  S \in \mathcal{B}_k\\
	\|\sum_{i \in S\cap A_k} u_i\|_2^2 & \leq &  \ 2\sum_{i\in S\cap A_k}\| u_i\|_2^2 \qquad   \textrm{for each }  S \in \mathcal{L}_k \label{sdp3}  \\
	\| \sum_{i\in S\cap A_k}x_{k-1}(i) u_i\|_2^2  & \leq & \ 2\sum_{i\in S\cap A_k}\| u_i\|_2^2 \qquad   \textrm{for each }  S \in \mathcal{L}_k \qquad \label{sdp4}  \\
   	\|u_i\|_2^2 & \le &  1  \qquad \forall i \in A_k \nonumber 
\end{eqnarray}
\item 
\label{apx:round}
Let $r_k \in \R^n$ be a random $\pm 1$ vector, obtained by setting each coordinate $r_k(i)$ independently to $-1$ or $1$ with probability $1/2$.

For each $i \in A_k$, update $x_k(i)=x_{k-1}(i)+\gamma\langle r_k, u_i\rangle$. 
For  each $i \not\in A_k$, set $x_k(i)=x_{k-1}(i)$.

\item
\label{apx:rnd}
 Initialize $A_{k+1}=A_{k}$.

For each $i$, if $|x_k(i)|\geq 1-1/n$, update $A_{k+1} = A_{k+1} \setminus \{i\}$.
\end{enumerate}
\item 
\label{stp3}
Generate the final coloring as follows.
For the frozen elements $i\notin A_{T+1}$, set $x_T(i)=1$ if $x_T(i)\geq 1-1/n$ and $x_T(i)=-1$ otherwise.
For the alive elements $i \in A_{T+1}$, 
set them arbitrarily to $\pm 1$. 
\end{enumerate}

Note that the SDP at time $k$ uses the vectors $u_i$ to generate the update $\Delta x_k(i)$ by projecting $u_i$ to the random vector $r_k$ and 
scaling this by $\gamma$. If we think of $u_i$ as one dimensional vectors (so $\Delta x_k(i) = \gamma r u_i$ where $r$ is randomly $\pm 1$), 
constraints \eqref{sdp1} will ensure that a set incurs zero discrepancy as long as it is big.
Constraints \eqref{sdp3} require the updates to satisfy the  proportional discrepancy property mentioned earlier.
Constraints \eqref{sdp4}  require the updates to satisfy the approximate orthogonality property mentioned earlier.

 \section{Analysis}
\label{s:analysis}

We begin with some simple observations.

\begin{lemma}
\label{lem:tri}
For any vector $u \in \R^n$ and a random vector $r \in \{\pm 1\}^n$,  $\E[\langle r,u\rangle^2] = \|u\|_2^2$   
and $|\langle r,u\rangle | \leq \sqrt{n}\|u\|_2$.
\end{lemma}
\begin{proof}
Writing $u$ in terms of its coordinates $u=(u(1),\ldots,u(n))$, 
\begin{align*}
\E[\langle r,u\rangle^2] =  \E [(\sum_i  r(i) u(i) )^2]  = \sum_{i,j} \E[r(i)r(j)] u(i)u(j)  = \|u\|_2^2 
\end{align*} 
where the last equality uses that $\E[r(i)r(j)]=0$ for $i \neq j$  and $\E[r(i)^2]=1$.

The second part follows by Cauchy-Schwarz inequality, as $|\langle r,u\rangle | \leq \|r\|_2 \|u\|_2 = \sqrt{n}\|u\|_2$.
\end{proof}

This implies the following.

\begin{obs}
The rounding of frozen elements in step \ref{stp3} of the algorithm affects the discrepancy of any set by at most $n \cdot (1/n)=1$. So we can ignore this rounding error.
Moreover, as $\|u_i\|_2 \leq 1$, $|\gamma \langle r,u_i \rangle |\leq \gamma\sqrt{n} \|u_i\|_2 \leq 1/n$, which implies that no $x_k(i)$ goes out of the range $[-1,1]$ during any step of the algorithm.
\end{obs}

The rest of the analysis is divided into three parts. In Section \ref{s:sdpfeasibility}, we show that the SDP is feasible and has value at least $|A_k|/3$ at each time step $k$. In Section \ref{s:disc}, we use the properties of the SDP to show that each set in $\mathcal{S}$ has discrepancy $O((t\log n)^{1/2})$ after $T$ steps with high probability. 
Finally, in Section \ref{s:terminate} we show that there are no alive elements after $T$ steps with high probability. Together these will imply Theorem \ref{thm:main}.

\subsection{SDP is feasible and has value $\Omega(|A_k|)$}
\label{s:sdpfeasibility}

 
To show that the SDP has value at least $|A_k|/3$ at any time step $k$, we will consider the dual and show that no solution with objective value less than $|A_k|/3$ can be feasible. 
By strong duality, this suffices as if the optimum (primal) SDP solution was less than $|A_k|/3$, there would also be some feasible dual solution with that value (provided Slater's conditions are satisfied). 

It might be useful to point out here that the feasibility of our SDP is incomparable to the main result in \cite{N13}; we can ensure a zero discrepancy to a few rows, which was not possible in the approach used in \cite{N13} but we can only ensure a partial colouring ($\sum_i \| u_i\|_2^2 \ge |A_k|/3$), whereas the SDP in \cite{N13} was feasible with the stronger constraint $\| u_i\|_2=1$ for all $i$.


To make  it easier to write the dual, we rewrite the SDP in  the following matrix notation by setting $X$ to be the Gram matrix of vectors corresponding to alive elements i.e. $X_{ij}=\langle u_i,u_j\rangle$ for $i,j\in A_k$. 
\begin{eqnarray*}
	\textrm{Maximize } I\bullet X && \quad \textrm{subject to}  \\
	v_{S}v_{S}^T\bullet X & = &  \ 0 \qquad \textrm{for each }  S \in \mathcal{B}_k \\
	(v_{S}v_{S}^T-2I_{S})\bullet X & \leq &  \ 0 \qquad \textrm{for each }  S \in \mathcal{L}_k \\
	(x_{S}x_{S}^T-2I_{S})\bullet X & \leq &  \ 0 \qquad  \textrm{for each }  S \in \mathcal{L}_k \\
	(e_ie_i^T)\bullet X & \leq &  \ 1 \qquad \forall i \in A_k \nonumber \\
   	X & \succeq & \ 0
\end{eqnarray*}

Here $v_{S}$ is the indicator vector of set $S\cap A_k$, $x_{S}$ is the vector with $i^{th}$ entry equal to $x_{k-1}(i)$ if $i\in S\cap A_k$ and $0$ otherwise and $I_{S}$ is the identity matrix restricted to set $S\cap A_k$, i.e.~$(I_{S})_{ii}=1$ if $i\in S\cap A_k$ and $0$ otherwise. $\bullet$ denotes the usual inner product on matrices $ A \bullet B = \textrm{Tr}(A^TB) = \sum_{ij} A_{ij}B_{ij}$.

We can write the dual of the above SDP (for reference, see \cite{GM12}), which is given by:
\begin{eqnarray}
	\textrm{Minimize } \sum_{i\in A_k} b_i \nonumber\\
	\textrm{s.t.} \qquad 
	\sum_{i\in A_k}  b_i e_ie_i^T  +   \sum_{S \in \mathcal{B}_k}\alpha_{S} v_{S}v_{S}^T & + & \sum_{S \in \mathcal{L}_k}\left(\beta_{S} (v_{S}v_{S}^T-2I_{S})+\beta_{S}^x(x_{S}x_{S}^T-2I_{S})\right)  \succeq   \ I \qquad
	  \label{dualconstr}    \\
	 b_i  &\geq &  \ 0 \qquad \forall i \in  A_k  \label{bval} \\
   	 \alpha_{S}  &\in &   \ \mathbb{R} \qquad \forall S \in \mathcal{B}_k \label{alphaval} \\
	 \beta_{S},\beta_{S}^x  &\geq &  \ 0 \qquad \forall  S \in \mathcal{L}_k  \label{betaval}
\end{eqnarray}
Here $A\succeq B$ denotes that the matrix $A-B$ is positive semi-definite. To show strong duality we use the following result. 


\begin{thm}[Theorem 4.7.1, \cite{GM12}]
\label{thm:duality}
If the primal program $(P)$ is feasible, has a finite optimum value $\eta$ and has an interior point $\tilde{x}$, then the dual program $(D)$ is also feasible and has the same finite optimum value $\eta$.
\end{thm}

\begin{lemma}
The SDP described above is feasible and has value equal to its dual program.
\end{lemma}
\begin{proof}
We apply Theorem~\ref{thm:duality}, with $P$ equal to the dual of the SDP. This would suffice as the dual $D$ of $P$ is our SDP. 

We claim that $b_i=1+\epsilon$ for $\epsilon>0$ for all $i\in A_k$, $\alpha_{S}=0$ for all $S\in\mathcal{B}_k$ and $\beta_S=\beta_S^x=\epsilon/(8n^2)$ for all $S\in\mathcal{L}_k$ is a feasible interior point for $P$.
Clearly, this solution is strictly in the interior of the constraints \eqref{bval}-\eqref{betaval}.
That $\eqref{dualconstr}$ is satisfied and has slack in every direction follows as the the number of sets $S$ can be at most 
$t |A_k| \leq tn \leq n^2$, and that for any vector $v$,  $vv^T$ is a rank one PSD matrix with eigenvalue $\|v\|^2_2 \leq n$, and thus all 
eigenvalues of $v_Sv_S^T - 2 I_S$ and $x_Sx_S^T-2I_S$ lie in the range $[-2, n]$. 

As this point has objective value at most $(1+\epsilon)n$ and since $b_i$ are non-negative, $P$ has a finite optimum value.
%
%
\end{proof}

We wish to show that any feasible solution to the dual must satisfy $\sum_i b_i \geq |A_k|/3$. To do this, we will show that there is a large subspace $W$ of dimension at least $|A_k|/3$ where the operator 
\begin{align*}
 \sum_{S \in \mathcal{B}_k}\alpha_{S} v_{S}v_{S}^T  + \sum_{S \in \mathcal{L}_k}\left(\beta_{S} (v_{S}v_{S}^T-2I_{S})+\beta_{S}^x(x_{S}x_{S}^T-2I_{S})\right)
 \end{align*}
is negative semidefinite. This would imply that to satisfy \eqref{dualconstr}, $b_i$'s have to be quite large on average. We first give two general lemmas.
 

\begin{lemma}
\label{col1matr}
Given an $h\times n$ matrix $M$ with columns $z_1,z_2,\dots,z_n$. If $\|z_i\|_2\le 1$ for all $i\in[n]$, then there exists a subspace $W$ of $\mathbb{R}^n$ satisfying:
\begin{enumerate}[i)]
\item $dim(W)\ge \frac{n}{2} $, and
\item $\forall y\in W$, $\|My\|_2^2 \le 2\|y\|_2^2$
\end{enumerate}
\end{lemma}
\begin{proof}
Let the singular value decomposition of $M$ be given by $M=\sum_{i=1}^n \sigma_i p_i q_i^T$, where $0\le \sigma_1\le\dots\le\sigma_n$ are the singular values of $M$ and $\{p_i : i\in [n]\}, \{q_i : i\in [n]\}$ are two sets of orthonormal vectors (if $h<n$, some $p_i$'s and the corresponding $\sigma_i$'s will be zero).
Then,
\[
\sum_{i=1}^n\sigma_i^2 =\textrm{Tr}[\sum_{i=1}^n \sigma_i^2 q_iq_i^T]=\textrm{Tr}[M^TM]=\sum_{i=1}^n \|z_i\|_2^2 \le n
\]
So at least $\lceil\frac{n}{2}\rceil$ of the squared singular values $\sigma_i^2$s have value at most $2$, and thus $\sigma_1\le\dots\le\sigma_{\lceil\frac{n}{2}\rceil} \le \sqrt{2}$.
Let $W=\textrm{span}\{q_1,\dots,q_{\lceil\frac{n}{2}\rceil}\}$. For $y\in W$,
\begin{eqnarray}
\|My\|_2^2 & = & \|\sum_{i=1}^n \sigma_ip_iq_i^T y\|_2^2 
  =   \|\sum_{i=1}^{\lceil\frac{n}{2}\rceil} \sigma_ip_iq_i^T y\|_2^2 \nonumber\\
 & \le & \sum_{i=1}^{\lceil\frac{n}{2}\rceil} \sigma_i^2 (q_i^Ty)^2  \quad \textrm{(since $p_i$ are orthonormal)} \nonumber\\
 &\le & 2\sum_{i=1}^{\lceil\frac{n}{2}\rceil}(q_i^Ty)^2 \nonumber \\
 &=& 2\|y\|_2^2 \qquad \textrm{(since $q_i$ are orthonormal)} \nonumber
\end{eqnarray}
\end{proof}

This implies the following result.

\begin{thm}
\label{bigsubspace}
Let $\mathcal{V}$ be any finite collection of vectors  $v_1,\ldots,v_{h}$ in $\mathbb{R}^n$, and 
for each $v \in \mathcal{V}$, there is some non-negative multiplier $\beta_v \geq 0$.
Consider the operator  \[ B = \sum_{v \in \mathcal{V}}  \beta_v \left(  vv^T   - 2 \sum_{i=1}^n \langle v,e_i \rangle^2 e_ie_i^T \right)\]
where $e_i$ are the standard basis of $\mathbb{R}^n$.
Then there exists a subspace $W$ of dimension at least $n/2$ such that  $\langle y,By \rangle \leq 0$ for every $y \in W$, or equivalently $y^TBy \leq 0$  for every $y \in W$.
\end{thm}
\begin{proof}
Let $v_i$ denote $\langle v,e_i\rangle$. We can express $y^TBy$ as 
\begin{align*}
B \bullet yy^T 
&=\sum_v  \beta_v\left(  vv^T \bullet yy^T  -  2 (\sum_i v_i^2 e_ie_i^T ) \bullet yy^T \right) \\
&=\sum_v \beta_v \left( (\sum_i v_i y_i)^2    -  2 \sum_i v_i^2 y_i^2 \right)  
\end{align*}

Construct a matrix $M$ with rows indexed by $v$ for each $v \in \mathcal{V}$ and columns indexed by $i\in [n]$. The entries of $M$ are given by $M_{v,i}=  \beta_v^{1/2} v_i$.
Then, we can write 
\[\sum_v \beta_v \left(\sum_i v_i y_i \right)^2  = \|M y\|_2^2.\]
For each  $i \in [n]$, define $\beta_i^2 = \sum_v \beta_v v_i^2$ as the squared $\ell_2$-norm of column $i$ of $M$, and let  $D$ be an $n\times n$ diagonal matrix with entries $D_{ii}= \beta_i$.
Then,
\begin{align}
\label{eq:exp1}
\sum_v \beta_v ( (\sum_i v_i y_i)^2    -  2 \sum_i v_i^2 y_i^2)  = \|My\|_2^2 - 2  \|Dy\|_2^2
\end{align}
Let $N \subseteq [n]$ be the set of coordinates with $\beta_i > 0$. We claim that it suffices to focus on the coordinates in $N$.
Let us first observe that if $i \notin N$, i.e.~$\beta_i^2=0$, then  we can set $y_i$ arbitrarily as \eqref{eq:exp1} is unaffected. As the directions $e_i$ for $i\in N$ are orthogonal to the directions in $[n]\setminus N$, it 
suffices to show that there is a $|N|/2$ dimensional subspace $W$ in $\textrm{span}\{e_i:i\in N\}$ such that $ \|My\|_2^2 - 2  \|Dy\|_2^2 \leq 0$ for each $y \in W$.
The overall subspace we desire is simply $W \oplus \textrm{span}\{e_i: i \in [n]\setminus N\}$ which has dimension $|N|/2 + (n-|N|) \geq n/2$.

So, let us assume that $N=[n]$ (or equivalently restrict $M$ and $D$ to columns in $N$), which gives us that $\beta_i > 0$ for all $i\in N$ and hence that $D$ is invertible.

Let $M'=M D^{-1} $. The squared $\ell_2$-norm of each column in $M'$ is $\sum_v \beta_v v_i^2/ D_{ii}^2$ which equals $1$, and by Lemma \ref{col1matr}, there is a subspace $W'$ of dimension at least $|N|/2$ such that $\|M'y'\|_2^2 \leq  2 \|y'\|_2^2$ for each $y' \in W'$.
Setting $y=D^{-1}y'$ gives
 \[ \|M y \|_2^2 =  \|M'y'\|_2^2 \leq 2 \|y'\|_2^2 = 2 \|Dy\|_2^2, \]  
and thus $W = \{D^{-1}y': y'\in W'\}$ gives the desired subspace since $\textrm{dim}(W)= \textrm{dim}(W')$. 
\end{proof}

Going back to the dual SDP, this gives the following.
\begin{lemma}
\label{smallsub}
Let  $B_k=\sum_{S \in \mathcal{L}_k} \left(\beta_{S} (v_{S}v_{S}^T-2I_{S})+\beta_{S}^x(x_{S}x_{S}^T-2I_{S})\right)$. Then, there exists a subspace $W \subseteq \mathbb{R}^{|A_k|}$ of dimension at least $|A_k|/2$ such that for all $y\in W$, $y^TB_ky\le 0$.
\end{lemma}
\begin{proof}
We apply Theorem  \ref{bigsubspace} with vectors $v$ as $v_S$ and $x_S$ for each small set $S \in \mathcal{L}_k$, with the multipliers $\beta_S$ and $\beta_{S}^x$. Then,
\begin{align*}
B & =   \sum_{S \in \mathcal{L}_k} [ \beta_{S} (v_{S}v_{S}^T- 2 \sum_{i\in A_k} \langle v_S,e_i\rangle^2 e_ie_i^T)+
 \beta_{S}^x(x_{S}x_{S}^T-2\sum_{i\in  A_k}  \langle x_S,e_i\rangle^2 e_ie_i^T)] \\
& =   \sum_{S \in \mathcal{L}_k} [\beta_{S} (v_{S}v_{S}^T-2I_{S})+
 \beta_{S}^x(x_{S}x_{S}^T-2 \sum_{i\in S \cap A_k} x_{k-1}(i)^2 e_ie_i^T)] \\
& \succeq  \sum_{S \in \mathcal{L}_k} \left(\beta_{S} (v_{S}v_{S}^T-2I_{S})+\beta_{S}^x(x_{S}x_{S}^T-2I_{S})\right) \\
&=  B_k 
\end{align*}
Here we use that $v_S$ is the indicator vector for set $S\cap A_K$ with entries  $\langle v_S,e_i \rangle=1$ iff $i \in S \cap A_k$ and thus, 
$\sum_{i\in A_k} \langle v_S,e_i \rangle^2 e_ie_i^T =  I_{S}.$
Similarly for the vectors $x_S$, $\langle x_S,e_i\rangle  = x_{k-1}(i)$ for $i\in S\cap A_k$ and $0$ otherwise. 
The last step uses that $ x_{k-1}(i)^2 \leq 1$ and thus
\[- 2\sum_{i\in S\cap A_k} x_{k-1}(i)^2 e_ie_i^T \succeq  -2 I_{S}. \] 
By Theorem \ref{bigsubspace}, there is a subspace $W$ with $\textrm{dim}(W)\geq  |A_k|/2$ such that 
$ y^T By \leq 0$ for each $y\in W$. As $B \succeq B_k$, it also holds that $y^T B_k y \leq 0$ for each $y \in W$. 
\end{proof}

We now come to the main theorem of this subsection.

\begin{thm}
\label{largedual}
At time step $k$, the dual program has value at least $|A_k|/3$.
\end{thm}
\begin{proof}
As element $i$ in $A_k$ appears in at most $t$ sets, the number of big sets $|\mathcal{B}_k|$ at time step $k$ is at most $|A_k|t/at=|A_k|/a$.
Let $W_1$ be the subspace orthogonal to 
$C=\textrm{span}\{v_S: S \in \mathcal{B}_k\}$. 
Clearly,  $\textrm{dim}(C)\le  |\mathcal{B}_k| \leq |A_k|/a$. 

Let $W_0$ be the subspace guaranteed by Lemma~\ref{smallsub} for matrix $B_k$ such that $\textrm{dim}(W_0)\ge  |A_k|/2$ and for all $y\in W_0$, $y^TB_ky\le 0$.
Define the subspace $W=W_1\cap W_0$.  Then,
\begin{align*}
\textrm{dim}(W) &\ge \textrm{dim}(W_0)-\textrm{dim}(C) 
\ge   |A_k|/2 -|A_k|/a = |A_k|/3 
\end{align*}

Let $P_W$ be the projection operator on 
the subspace $W$. Projecting the dual constraint (\ref{dualconstr}) on to $W$, we get
\[ P_W \left( \sum_{i\in A_k} b_ie_ie_i^T  +\sum_{S \in \mathcal{B}_k}\alpha_{S}v_{S}v_{S}^T+ B_k \right) \succeq   \ P_W	 \]
By linearity of  $P_W$  and as $P_W(v_{S}v_{S}^T) = 0$ for each $S\in \mathcal{B}_k$, this implies 

\[ P_W \left(\sum_{i\in A_k} b_ie_ie_i^T \right)+P_W(B_k ) \succeq P_W \]

Taking trace on both the sides and noting that $\textrm{Tr}[P_W(B_k)]\le 0$ since $y^T B_ky\le 0$ for all $y\in W$, we get
\[\textrm{Tr}\left[P_W \left(\sum_{i\in A_k} b_ie_ie_i^T\right)\right] \ge  \textrm{Tr}[P_W]=\textrm{dim}(W)\]

As $b_i\ge 0$ for all $i\in A_k$, the operator $\sum_{i\in A_k} b_ie_ie_i^T$ is  positive semi-definite. As taking the projection of a positive semi-definite operator can only decrease its trace, we can lower bound the dual objective as
\begin{align*}
\sum_{i\in A_k}b_i  & = \textrm{Tr}\left[\sum_{i\in A_k} b_ie_ie_i^T \right]
\ge \textrm{Tr}\left[P_W\left(\sum_{i\in A_k} b_ie_ie_i^T \right)\right] \ge  \textrm{Tr}[P_W]=\textrm{dim}(W) \ge |A_k|/3 
\end{align*}
which completes the proof.
\end{proof}

\subsection{Bounding the discrepancy}
\label{s:disc}
Let $D_S(k)$ denote the signed discrepancy of set $S\in\mathcal{S}$ at end of time step $k$ i.e.~$D_S(k)=\sum_{i\in S}x_k(i)$.

\noindent
We now show the following key result.
\begin{thm}
\label{thm:discrepancy}
Fix a set $S\in\mathcal{S}$.
Then, for any $\lambda \geq 0 $, the discrepancy of $S$ at time step $T$ satisfies 
\[ \Pr \left[ |D_S(T)| \geq \lambda\sqrt{t} \right] \leq 8 \exp(-\lambda^2/(100a)).\]
\end{thm}

Setting $\lambda = O(\log^{1/2}n)$ would imply that with high probability every set has discrepancy $O((t \log n)^{1/2})$ at time $T$. 

Among other things, the proof of Theorem \ref{thm:discrepancy} will use a powerful concentration inequality for martingales due to Freedman that we describe below.

\paragraph*{Martingales and Freedman's inequality}
Let $X_1,X_2,\ldots,X_n$ be a sequence of independent random variables on some probability space, and let 
$Y_k$ be a function of $X_1,\ldots,X_k$. The sequence 
$Y_0,Y_1,Y_2,\ldots,Y_n$ is called a martingale with respect to the sequence $X_1,\ldots,X_n$ if for all $k \in [n]$, 
$\E[|Y_k|]$ is finite and $\E[Y_k| X_1,X_2,...,X_{k-1}]= Y_{k-1}$. We will use $\E_{k-1}[Z]$ to denote $\E[Z | X_1,X_2,...,X_{k-1}]$ where $Z$ is any random variable.

\begin{thm}[Freedman \cite{Freedman}]
\label{thm:freedman}
Let $Y_0,\ldots,Y_n$ be a martingale with respect to $X_1,\ldots,X_n$ such that 
$|Y_k - Y_{k-1}| \leq M$ for all $k$, and let 
\begin{align*}
W_k &= \sum_{j=1}^k \E_{j-1}[(Y_j - Y_{j-1})^2] = \sum_{j=1}^k\textrm{Var}[Y_j|X_{1},\ldots,X_{j-1}].
\end{align*}
Then for all $\lambda \geq 0$  and $\sigma^2 \geq 0$, we have 
\begin{align*}
 \Pr[|Y_n - Y_0| \geq \lambda \textrm{ and } W_n \leq \sigma^2] 
 \le 2 \exp\left(-\frac{\lambda^2}{2 (\sigma^2 + M \lambda/3)} \right).  
\end{align*}
\end{thm}

Observe crucially that the above inequality is much more powerful than the related Azuma-Hoeffding or
Bernstein's inequality. In particular, the term $W_n$ is the variance encountered by the martingale on the particular sample path it took, as opposed to a worst case bound on the variance over all possible paths.

\paragraph*{Simple Observations}
We now get back to the proof of Theorem \ref{thm:discrepancy} and begin with a few simple observations.

Fix a set $S\in \mathcal{S}$. 
Let the vector solution returned by the SDP at time $k$ be given by vectors $u_i^k$ for $i\in [n]$ where we take $u_i^k=0$ if $i\not\in A_k$. 
We say that 
$S$ becomes {\em active} at time $k$ if $k$ is the first time step when $|S\cap A_k|\le at$.

\begin{obs}
\label{obs:onlysmall}
Before a set $S\in\mathcal{S}$ is active, it incurs zero discrepancy.
\end{obs}
\begin{proof}
Suppose $S$ becomes active at time $k_S$. Then, $D_S(k_S-1)=\sum_{k=1}^{k_S-1}\gamma\langle r_k, \sum_{i\in S\cap A_k}u_i^k\rangle=0$, since by  SDP constraint (\ref{sdp1}), $\sum_{i\in S\cap A_k}u_i^k=0$ for $k<k_S$.
\end{proof}

\begin{obs}
\label{obs:smallt}
As a set has no more than $at$ alive variables when it becomes active, Observation \ref{obs:onlysmall} implies that the maximum discrepancy any set can have is $2at$, which gives Theorem \ref{thm:discrepancy} for $\lambda > 2at^{1/2}$. 
So henceforth we can assume that $\lambda \leq 2a t^{1/2}$.
\end{obs}

Define the {\em energy} of  set $S$ at end of time step $k$ as $E_S(k)=\sum_{i\in S}x_k(i)^2$ and change in energy of $S$ at time step $k$ as $\Delta_k E_S = E_S(k)-E_S(k-1)$. Then, 
\begin{eqnarray}
\Delta_kE_S &= & \sum_{i\in S}x_k(i)^2-\sum_{i\in S}x_{k-1}(i)^2 = \sum_{i\in S} \left((x_{k-1}(i)+\gamma\langle r_k,u_i^k\rangle)^2-x_{k-1}(i)^2 \right) \nonumber \\
& =& \gamma^2\sum_{i\in S} \langle r_k,u_i^k\rangle^2+2\gamma\langle r_k,\sum_{i\in S}x_{k-1}(i)u_i^k\rangle \label{eq:inject}
\end{eqnarray}
The following is a simple but crucial observation.

\begin{obs}
\label{obs:energy-increase}
Once a set $S\in\mathcal{S}$ becomes active, its energy can increase overall by at most $at$. 
\end{obs}
\begin{proof}
When $S$ becomes active, it has at most $at$ alive variables. Moreover, a frozen variable is never updated by the algorithm and can never become alive again. As the energy of a single variable is bounded by $1$, the energy of $S$ can increase by at most $at$ after it becomes active.
\end{proof}

{\em Remark:} Note that the energy of a set $S$ does not necessarily increase monotonically over time. It evolves randomly and can also decrease. So, even though the overall increase is at most $at$, the total energy ``injected"
$\sum_{k\geq k_S} |\Delta_k E_S|$ can be arbitrarily larger than $at$. Here $k_S$ denotes the time when $S$ becomes active.

\medskip

By Observation~\ref{obs:onlysmall}, we only need to bound the discrepancy of $S$ after it becomes active. For notational convenience, let us call the time $S$ becomes active as time $0$. So, $D_S(k)$
and $E_S(k)$ will be the signed discrepancy and energy of $S$ respectively, $k$ time steps after it becomes active.

\begin{obs}      
\label{obs:var}
After $S$ becomes active, $D_S(k)$ behaves like a martingale with variance of increment at time step $k$  bounded by
\[ \E_{k-1}[ (D_S(k)-D_S(k-1))^2]  \leq  2 \gamma^2 \sum_{i\in S\cap A_k} \|u_i^k \|_2^2 \]
\end{obs} 
\begin{proof}
The discrepancy update of $S$  at time $k$ is 
$\gamma\langle r_k,\sum_{i\in S\cap A_k}u_i^k\rangle$. This has expectation $0$ (averaging of $r_k$) and by Lemma \ref{lem:tri} this  variance is exactly $\gamma^2 \| \sum_{i\in S\cap A_k} u_i^k \|_2^2$, which 
by SDP constraint (\ref{sdp3}) is upper bounded by $  2 \gamma^2 \sum_{i\in S\cap A_k} \|u_i^k \|_2^2$.
\end{proof}

\medskip

\paragraph*{ Proof of Theorem \ref{thm:discrepancy}}
The plan of the proof is the following. Freedman's inequality allows us to bound the discrepancy at time $T$ as a function of the variance  
$\sum_{k=1}^T \E_{k-1}[ (D_S(k)-D_S(k-1))^2]$ which is at most 
 $2 \gamma^2 \sum_{k=1}^T \sum_{i\in S\cap A_k} \|u_i^k \|_2^2$ by Observation \ref{obs:var}.
As we will see, this term is
expected total energy injected into $S$. 

As the overall energy increase of $S$ can be at most $at$ (Observation \ref{obs:energy-increase}), it would suffice to show that the total injected energy into $S$ is comparable to $at$.
To do this, we will use the approximate orthogonality constraints \eqref{sdp4} and apply Freedman's inequality again to show that the injected energy is tightly concentrated around the energy increase. 
We now give the details. 

\medskip

Recall that by \eqref{eq:inject}, the energy change at time $k$ is a random variable given by 
\[
\Delta_kE_S = \gamma^2\sum_{i\in S} \langle r_k,u_i^k\rangle^2+2\gamma\langle r_k,\sum_{i\in S}x_{k-1}(i)u_i^k\rangle \]
Denote the first term above as \[\Delta_kQ_S= \gamma^2\sum_{i\in S} \langle r_k,u_i^k\rangle^2\] which we will call the change in quadratic energy of $S$ at time step $k$ and let $Q_S(k)=\sum_{j=1}^k\Delta_jQ_S$, the total quadratic energy of $S$ till time $k$. 

 Similarly, denote the second term as 
\[ \Delta_kL_S=2\gamma\langle r_k,\sum_{i\in S}x_{k-1}(i)u_i^k\rangle\]
which we will call the change in linear energy of $S$ at time step $k$, and let $L_S(k)=\sum_{j=1}^k\Delta_jL_S$, the total linear energy of $S$ till time $k$. 
The energy of $S$ at time $k$ is given by $E_S(k)=Q_S(k)+L_S(k)$. 

Define $Q'_S(k)$ as
\[Q'_S(k) =\sum_{j=1}^k\E_{j-1}[\Delta_jE_S]=\sum_{j=1}^k\E_{j-1}[\Delta_jQ_S].\]
 By lemma~\ref{lem:tri},
\[
Q'_S(k)
=\sum_{j=1}^k\gamma^2\sum_{i\in S}\|u_i^j\|_2^2
\]

We are now ready to prove the tail bound on discrepancy.
The probability that discrepancy of $S$ at time $T$ exceeds $\lambda\sqrt{t}$ can be written as
\begin{align}
\Pr\left[|D_S(T)| \ge \lambda\sqrt{t}\right] \leq    \Pr\left[|D_S(T)| \ge \lambda\sqrt{t}, Q'_S(T)\le 16at \right] + \Pr \left[Q'_S(T) > 16 at \right]\label{eq:prob}
\end{align}
We now bound each of the terms in \eqref{eq:prob} separately.

\medskip
\noindent
{\bf Bounding the first term.} Recall that $D_S(k)$ is a martingale. To apply Freedman's inequality(Theorem~\ref{thm:freedman}) we bound $M$ and $W_k$ as follows. By Lemma~\ref{lem:tri},
\begin{align*}
M&\le|D_S(k)-D_S(k-1)| =|\gamma\langle r_k,\sum_{i\in S}u_i^k\rangle| \le \gamma\sqrt{n}\|\sum_{i\in S}u_i^k\|_2 \le \gamma n^{3/2}
\end{align*}
Similarly from Lemma~\ref{lem:tri} and the SDP constraint (\ref{sdp3}),
\begin{eqnarray*}
W_k & =&\sum_{j=1}^k\E_{j-1}[(D_S(j)-D_S(j-1))^2]  
=\sum_{j=1}^k \E_{j-1}[\gamma^2\langle r_j,\sum_{i\in S}u_i^j\rangle^2] \\
\label{seqn1} &=& \sum_{j=1}^k \gamma^2 \| \sum_{i\in S}u_i^j\|_2^2 
\le  \sum_{j=1}^k 2\gamma^2 \sum_{i\in S} \|u_i^j\|_2^2 = 2Q'_S(k)
\end{eqnarray*}
Freedman's inequality now gives,
\begin{eqnarray}
\Pr \left[|D_S(T)| \ge \lambda\sqrt{t} \textrm{ and } Q'_S(T)\le 16at \right] 
& \le&  \Pr \left[|D_S(T)| \ge \lambda\sqrt{t} \textrm{ and } W_T\le 32at \right] \nonumber \\
&\le & 2\exp\left(\frac{-\lambda^2t}{2[32at+\gamma n^{3/2}\lambda\sqrt{t}/3]}\right) \nonumber \\
\label{freeduse} & \le & 2\exp\left(\frac{-\lambda^2}{100a}\right)   \qquad \textrm{(using $\lambda\le 2a\sqrt{t}$)}
\end{eqnarray}

\medskip
\noindent
{\bf Bounding the second term.} We can write 
\begin{align}
&  \Pr\left[Q'_S(T) > 16 at \right]  
=   \sum_{j=0}^\infty \Pr \left[ 2^{j+4}at<Q'_S(T)\le 2^{j+5}at \right]  \nonumber\\
&\le  \Pr\left[Q_S(T)\le Q'_S(T)-8at \right] + \sum_{j=0}^\infty \Pr \left[Q'_S(T)\le 2^{j+5}at, Q_S(T)\ge 2^{j+3}at\right]
\label{line2}
\end{align}
The inequality above holds as the event $\{Q'_S(T) > 16 at\}$ is contained in the union of the two events in \eqref{line2}.

As the energy $E_S(T)$ of $S$ cannot exceed $at$, we have $E_S(T) = L_S(T)+Q_S(T)\le at$. Thus,
 $Q_S(T)\ge 2^{j+3}at$ implies $L_S(T) \le at -2^{j+3}at \le -7 \cdot2^j at$, giving
\begin{align}
\Pr[Q'_S(T) > 16 at]  \leq  \Pr[Q_S(T)\le Q'_S(T)-8at] + \sum_{j=0}^\infty  \Pr[L_S(T)\le -7\cdot2^jat, Q'_S(T)\le 2^{j+5}at]
\label{eqn:probab}
\end{align}

To bound the second term on the right hand side of (\ref{eqn:probab}), we will crucially use the approximate orthogonality constraints in the SDP (\ref{sdp4}) and use Freedman's inequality.
To this end, note that $L_S(k)$ is a martingale whose difference sequence can be bounded by 
\begin{align*}
M\le|L_S(k)-L_S(k-1)| =|2\gamma\langle r_k,\sum_{i\in S}x_{k-1}(i)u_i^k\rangle| \le 2\gamma\sqrt{n}\|\sum_{i\in S}x_{k-1}(i)u_i^k\|_2 \le 2\gamma n^{3/2}
\end{align*}
where we used Lemma~\ref{lem:tri} in the first inequality and the fact that $|x_{k-1}(i)|\le 1$. 

By Lemma~\ref{lem:tri} and SDP constraint (\ref{sdp4}),
\begin{eqnarray}
W_k & =&\sum_{j=1}^k\E_{j-1}[|L_S(j)-L_S(j-1)|^2] =\sum_{j=1}^k \E_{j-1}[4\gamma^2\langle r_j,\sum_{i\in S}x_{j-1}(i)u_i^j\rangle^2] \nonumber \\
&=& \sum_{j=1}^k 4\gamma^2 \| \sum_{i\in S}x_{j-1}(i)u_i^j\|_2^2 \le  \sum_{j=1}^k 8\gamma^2 \sum_{i\in S} \|u_i^j\|_2^2 = 8Q'_S(k)  \nonumber
\end{eqnarray}
Applying Freedman's inequality now with these bound on $M$ and $W_k$, we obtain
\begin{align}
\Pr\left[|L_S(T)| \ge 7\cdot2^jat \textrm{ and } Q'_S(T)\le 2^{j+5}at\right]
& \le  \Pr\left[|L_S(T)| \ge 7\cdot2^jat\textrm{ and } W_T\le 2^{j+8}at\right] \nonumber \\
&\le 2\exp\left(\frac{-49\cdot2^{2j}a^2t^2}{2[2^{j+8}at+2\gamma n^{3/2} \cdot 7\cdot2^jat/3]}\right) \nonumber \\
& \le 2\exp\left(\frac{-2^{j}at}{20}\right)  \nonumber
\end{align}

Together with $\lambda \leq 2a\sqrt{t}$ (by our assumption), this gives  
\begin{align}
\sum_{j=0}^\infty  \Pr[L_S(T)\le -7\cdot2^jat, Q'_S(T) 
\le 2^{j+5}at] \leq 4 \exp\left(\frac{-at}{20}\right) \leq 4 \exp(\frac{-\lambda^2}{100a}) \label{freeduse2}
\end{align}

It remains to bound $\Pr[Q_S(T)\le Q'_S(T) -8at]$, the first term in \eqref{eqn:probab}. 
We use Freedman's inequality in a simple way (even Azuma-Hoeffding would suffice here). 

Define the martingale $Z_k=Q_S(k)-\sum_{j=1}^k\E_{j-1}[\Delta_jQ_S]=Q_S(k)-Q'_S(k)$ (this is the standard Doob decomposition of $\Delta_kQ_S$). By Lemma~\ref{lem:tri},
\begin{align*}
M\le|Z_k-Z_{k-1}| =|\Delta_kQ_S-\E_{k-1}[\Delta_kQ_S]| 
\le 2|\Delta_kQ_S| \le 2\gamma^2n\sum_{i\in S} \|u_i^k\|_2^2 \le 2\gamma^2n^2
\end{align*}
Using the trivial bound $\E_{j-1}[(Z_j-Z_{j-1})^2]\le M^2$, we obtain that 
\begin{align*}
W_T &=\sum_{j=1}^T \E_{j-1}[(Z_k-Z_{k-1})^2] \le 4T\gamma^4n^4 = 48 \gamma^2 n^4 \log n.
\end{align*}
As $Q_S(T)\le Q'_S(T)-8at $ is the same as $Z_T\le -8at$, by Freedman's inequality we get
\begin{align}
 \Pr[Q_S(T)\le Q'_S(T)-8at]  &\le    \Pr[|Z_T|\ge 8at] \nonumber   \\
& \le    2\exp\left(\frac{-64a^2t^2}{2[W_T+16\gamma^2n^2at/3]}\right)  \nonumber \\
& \le    2\exp\left(\frac{-a^2t^2}{2\gamma^2 n^4 \log n}\right) \nonumber\\
& \leq 2\exp(-a^2t^2) \leq 2\exp(\frac{-\lambda^2}{100a})
\label{finally} 
\end{align}
In the last step we use that $\gamma = 1/(n^2 \log n)$, and Observation \ref{obs:smallt}.

Combining equations (\ref{eq:prob}),(\ref{freeduse}),(\ref{eqn:probab}),(\ref{freeduse2}) and (\ref{finally}), we obtain the desired bound
\begin{eqnarray*}
 \Pr[|D_S(T)| \ge \lambda\sqrt{t}]  & \le & 8\exp\left(\frac{-\lambda^2}{100a}\right)
\end{eqnarray*}

\subsection{Termination and finishing the proof}
\label{s:terminate}
To finish the proof, we show that the last rounding step at time $T+1$ does not cause problems.

\begin{thm}
\label{thm:terminate}
After time $T$, there are no alive variables left with probability at least $1-O(n^{-2})$.
\end{thm}
\begin{proof}
Given the coloring $x_k$ at time $k$, define $G_k = \sum_{i \in A_k} (1-x_k(i)^2)$. Clearly $G_1\leq n$. 
As $x_{k}(i) =   x_{k-1}(i) + \gamma \langle r_k,u_i^k \rangle $, we have that $\E_{k-1}[x_k(i)^2] = x_{k-1}(i)^2 + \gamma^2 \|u_i^k\|_2^2$.
It follows
\begin{eqnarray*} 
\E_{k-1}[G(k)]  & = &   \E_{k-1}\left[\sum_{i \in A_k} (1-x_{k}(i)^2)\right] 
                        =  \sum_{i \in A_k} \left(1-x_{k-1}(i)^2 \right) - \gamma^2 \sum_{i \in A_k} \|u_i^k\|_2^2   \\
                  & \leq &  \sum_{i \in A_k} \left(1-x_{k-1}(i)^2 \right) - \gamma^2  |A_k|/3  
                   \leq  (1 - \gamma^2/3)  \sum_{i \in A_k} (1-x_{k-1}(i)^2)  \\ 
                  & \leq &  (1 - \gamma^2/3)  \sum_{i \in A_{k-1}} (1-x_{k-1}(i)^2) 
                  =  (1 - \gamma^2/3)G_{k-1}  
\end{eqnarray*}
Thus by induction, 
\begin{align*}
\E[G_{T+1}] &\leq (1-\gamma^2/3)^T G_1 \leq e^{-\gamma^2 T/3 } n 
= n^{-4} \cdot n = 1/n^3.
\end{align*}
Thus by Markov's inequality, $\Pr[G_{T+1} \geq 1/n] \leq 1/n^2$. However, $G_{T+1} \leq 1/n$ implies that $A_{T+1}=0$ as each alive variable contributes at least $1-(1-1/n)^2 > 1/n$ to $G_{T+1}$.
\end{proof}

Theorem \ref{thm:main} now follows directly. 
Applying Theorem~\ref{thm:discrepancy} with $\lambda=  c \log^{1/2} n$ for $c$ a large enough constant and taking a union bound over the at most $nt \leq n^2$ sets, 
we get that $|D_S(T)|=O((t\log n)^{1/2})$ with probability at least $1-1/\textrm{poly}(n)$ for all sets $S$. 
By Theorem~\ref{thm:terminate} with probability at least $1-O(n^{-2})$, all variables are frozen by time $T$ and hence at most an additional discrepancy of $1$ is added by rounding the frozen variables to $\pm 1$.

%

{\hfill$\qed$}

\section{Extension to the Koml\'{o}s setting}
\label{s:komlos}


The algorithm also extends to the more general Koml\'{o}s setting with some additional modifications.
Recall that in the Koml\'{o}s setting, we are given an $m\times n$ matrix $B$ with arbitrary real entries $b_{ji}$ such that 
for each column $i$, it holds that $\sum_j b_{ji}^2 \leq 1$. Let $r_j$ denote the $j$-th row of $B$ and let $a$ be the constant as in the previous section.
We will show the following result.

\begin{thm}
\label{thm:discrepancykomlos}
Fix any row $r_j$ of matrix $B$.
Then, for any $\lambda \geq 0 $, the discrepancy of $r_j$ at time step $T$ (the end of the algorithm) satisfies 
\[ \Pr \left[ |D_T(r_j)| \geq \lambda \right] \leq 8 \exp(-\lambda^2/(1000a))\]
where $|D_T(r_j)|$ is the discrepancy of row $j$ after time step $T$.
\end{thm}

The previous argument does not work directly when the entries $b_{ji}$ are arbitrary
as we may not get strong concentration if some entries $b_{ji}$ are too large.
So we consider the following modified algorithm.

\paragraph{Algorithm.}

Given a matrix $B$, for any $\lambda >0$ we denote by $r_j^{\lambda}$ the $\lambda$-truncation of row $j$ containing only the entries $b_{ji}$ that are at most $4a/\lambda$ in absolute value i.e.,~$r_j^{\lambda}$ only contains those entries $i$ of row $j$ for which $|b_{ji}|\le 4a/\lambda$ and is $0$ otherwise.

As previously, let $A_k$ denote the set of alive variables at beginning of time step $k$, and we set $\gamma=1/n^6$ and $T=(12/\gamma^2)\log n$. 
A row $j$ is called {\em big} at time step $k$ if $\sum_{i\in A_k} b_{ji}^2>a$, and {\em small} otherwise.
As the $\ell_2$-norm of columns of $B'$ is at most $1$, there at most $|A_k|/a$ big rows at any time step $k$.

\paragraph{The modified SDP.}
The SDP is modified as follows: Similar to (\ref{sdp1}) we still require the discrepancy of big rows to be zero. That is,
\begin{equation}
\label{sdp1'} \|\sum_{i \in A_k}  b_{ji} u_i\|_2^2  =   \ 0 \qquad  \textrm{for each big row }  j 
\end{equation}

For an active (not big) row $r_j$ at time $k$, we add proportional discrepancy and approximate orthogonality constraints {\em for every $\lambda$-truncation $r_j^{\lambda}$ of $r_j$} i.e.,~for every $\lambda > 0$, we add the proportional discrepancy constraint \eqref{sdp3} (same as before, we just multiply the $u_i$'s by $b_{ji}$'s)
\begin{equation}
\label{sdp3'}
\| \sum_{i \in A_k, |b_{ji}|\le 4a/\lambda} b_{ji} u_i\|_2^2   \leq  2 \sum_{i\in A_k,|b_{ji}|\le 4a/\lambda} b_{ji}^2 \| u_i\|_2^2 
\end{equation}

and the approximate orthogonality constraints \eqref{sdp4}
\begin{equation}
\label{sdp4'}
\| \sum_{i\in A_k,|b_{ji}|\le 4a/\lambda} b_{ji}^2 x_{k-1}(i) u_i\|_2^2   \leq   2 \sum_{i\in A_k,|b_{ji}|\le 4a/\lambda} b_{ji}^4\| u_i\|_2^2. 
\end{equation}

Notice that as stated, for each active row we add two SDP constraints for every value of $\lambda > 0 $. However it suffices to add at most $2n$ constraints in total for each active row: just sort the entries of a row in increasing order of absolute value and add the proportional discrepancy and orthogonality constraints in the SDP for every prefix of this sorted row  (alternately, one could also consider geometrically increasing values of $\lambda$). Thus the SDP has a polynomial number of constraints at any time step.

\paragraph{Analysis.}

First, exactly as before the SDP is feasible and has a solution with value at least $|A_k|/3$. This follows from 
Theorem \ref{bigsubspace}, which shows that there is a subspace $W$ of dimension
at least $|A_k|/2$ where the corresponding operator is negative semidefinite on $W$,
and then applying the argument in Theorem \ref{largedual}. In fact, this would be true even if \eqref{sdp4'} was replaced by the stronger constraint
\[ \| \sum_{i\in A_k,|b_{ji}|\le 4a/\lambda} b_{ji}^2 x_{k-1}(i) u_i\|_2^2   \leq   2 \sum_{i\in A_k,|b_{ji}|\le 4a/\lambda} b_{ji}^4 x_{k-1}(i)^2 \| u_i\|_2^2. \]

Let $D_k(r_j)$ denote the signed discrepancy of row $j$ at the end of time step $k$, 
\[D_k(r_j)=\sum_{i\in[n]}b_{ji}x_k(i) . \] 
We also extend this definition to truncations of rows: 
\[D_k(r_j^\lambda)=\sum_{i\in[n],|b_{ji}|\le 4a/\lambda}b_{ji}x_k(i) . \]

We now show Theorem \ref{thm:discrepancykomlos}.
Fix a row $r_j$ and a $\lambda\ge 0$. Call an entry $b_{ji}$ large if $|b_{ji}|>4a/\lambda$.
We first make the following key observation.
\begin{obs} 
\label{obs:large}
When a row becomes active, the $\ell_1$-norm of the alive variables in that row that are large can be at most $\lambda/4$.
\end{obs}
\begin{proof}
Each large entry is at least $4a/\lambda$ in absolute value. As a row $r_j$ becomes active when $\sum_{i \in A_k} b_{ji}^2 \leq a$, there can be at most $a/(4a/\lambda)^2 = \lambda^2/16a$ alive variables with $b_{ji}$ large. By Cauchy-Schwarz inequality, 
\[ \left(\sum_{i \in A_k,|b_{ji}|> 4a/\lambda}  |b_{ji}| \right) \leq \left(\sum_{i \in A_k} b_{ji}^2 \right)^{1/2} \left(\frac{\lambda^2}{16a}\right)^{1/2} \leq \lambda/4. \]
\end{proof}
The above observation implies that when a row becomes active, the large entries in it can change discrepancy by at most $\lambda/2$.
Thus to prove Theorem~\ref{thm:discrepancykomlos}, it suffices to show
\[ \Pr \left[ |D_T(r_j^\lambda)| \geq \lambda/2 \right] \leq 8 \exp(-\lambda^2/(1000a)).\]

This follows similarly to the analysis as before, using the proportional discrepancy (\ref{sdp3'}) and approximate orthogonality constraints (\ref{sdp4'}) for $r_j^\lambda$ and noting that (\ref{sdp4'}) implies that
\begin{equation}
\label{eqn:b4b2}
\| \sum_{i\in A_k,|b_{ji}|\le 4a/\lambda} b_{ji}^2 x_{k-1}(i) u_i\|_2^2   \leq  \sum_{i\in A_k,|b_{ji}|\le 4a/\lambda} b_{ji}^4\| u_i\|_2^2  \leq  \frac{32a^2}{\lambda^2} \sum_{i\in A_k,|b_{ji}|\le 4a/\lambda} b_{ji}^2\| u_i\|_2^2 
\end{equation} 
as $|b_{ji}| \le 4a/\lambda$ for all entries in the truncated row $r_j^\lambda$.
Let us define the energy of $\lambda$-truncation of row $j$ at time $k$ as  
\[E_k(r_j^\lambda)=\sum_{i\in[n], |b_{ji}|\le 4a/\lambda} b_{ji}^2 x_k(i)^2.\]
As previously, once the row becomes active, its energy can rise by at most $a$. 

\smallskip

The analysis in Section \ref{s:disc} had two main ideas:
\begin{enumerate}
\item 
First we showed that the expected squared discrepancy of a set $S$ at time $T$ was $O(1)$ times the energy injected into the set $Q'_S(T)$ (using constraints \eqref{sdp3}).
This argument works exactly  as before using constraints \eqref{sdp3'} and we sketch the details below.

\medskip

For ease of notation we will denote the entries of the truncated row $r_j^\lambda$ as $b_{ji}$ where it is understood that we are setting $b_{ji}=0$ if $b_{ji}$ was large in the original matrix. 
The change in energy at time $k$ is a random variable given by 
\[
\Delta_kE(r_j^\lambda) = \gamma^2\sum_{i\in [n]} b_{ji}^2\langle r_k,u_i^k\rangle^2+2\gamma\langle r_k,\sum_{i\in [n]}b_{ji}^2x_{k-1}(i)u_i^k\rangle \]
Denote the first term above as $\Delta_kQ(r_j^\lambda)$, the change in quadratic energy of $r_j^\lambda$ at time step $k$ and let $Q_k(r_j^\lambda)=\sum_{k'=1}^k\Delta_{k'}Q(r_j^\lambda)$, the total quadratic energy of $r_j^\lambda$ till time $k$. 

 Similarly, denote the second term as 
$ \Delta_kL(r_j^\lambda)$, the change in linear energy of $r_j^\lambda$ at time step $k$, and let $L_k(r_j^\lambda)=\sum_{k'=1}^k\Delta_{k'}L(r_j^\lambda)$, the total linear energy of $r_j^\lambda$ till time $k$. 

Define $Q'_k(r_j^\lambda) =\sum_{k'=1}^k\E_{k'-1}[\Delta_{k'}E(r_j^\lambda)]=\sum_{k'=1}^k\E_{k'-1}[\Delta_{k'}Q(r_j^\lambda)]$.  By lemma~\ref{lem:tri},
\[
Q'_k(r_j^\lambda)
=\sum_{k'=1}^k\gamma^2\sum_{i\in [n]}b_{ji}^2\|u_i^{k'}\|_2^2
\]

Just as before, discrepancy $D_k(r_j^\lambda)$ behaves as a martingale with the variance $W_k$ bounded by $2Q'_k(r_j^\lambda)$. 
Freedman's inequality then gives,
\begin{eqnarray}
\label{komloseqn1}
\Pr \left[|D_T(r_j^\lambda)| \ge \lambda/2 \textrm{ and } Q'_T(r_j^\lambda)\le 16a\right] & \le & 2\exp\left(\frac{-\lambda^2}{1000a}\right) 
\end{eqnarray}

Next we showed that $Q'_S(T)$ was essentially the same as $Q_S(T)$ (shown in \eqref{finally}). In fact this difference can be made arbitrarily small by reducing $\gamma$ and the argument works exactly as before here. In particular, we get
\begin{eqnarray}
\label{komloseqn2}
 \Pr[Q_T(r_j^\lambda)\le Q'_T(r_j^\lambda)-8a]  \le 2\exp\left(\frac{-\lambda^2}{1000a}\right)
\end{eqnarray}

\item 
The second part was to show that the linear term does not cause problems. In particular,
the crucial argument was that $Q_S(T)$ cannot be much more than $at$ as (i) the total rise in energy $L_S(T)+ Q_S(T)$ cannot exceed $at$ and (ii)
$L(T)$ was a martingale with squared deviation comparable to $Q_S(T)$  and hence cannot be much larger than $Q_S^{1/2}(T)$. This step used the constraints \eqref{sdp4}.

 This argument also works similarly in our setting here. For a truncated row $r_j^\lambda$, $Q_T(r_j^\lambda)$ cannot be much more than $a$ as (i) the total rise in energy $L_T(r_j^\lambda)+ Q_T(r_j^\lambda)$ cannot exceed $a$ and (ii)
$L_T(r_j^\lambda)$ is a martingale with squared deviation comparable to $\frac{32a^2}{\lambda^2}Q_T(r_j^\lambda) $ (by \eqref{sdp4'} and \eqref{eqn:b4b2}).
Proceeding exactly as before and applying Freedman's inequality we obtain that, 
\begin{equation}
\label{komloseqn3}
\sum_{\ell=0}^\infty  \Pr\left[L_T(r_j^\lambda)\le -7\cdot2^{\ell} a, \quad Q'_T(r_j^\lambda)\le 2^{\ell+5}a\right] \leq 4 \exp(\frac{-\lambda^2}{1000a}).
\end{equation}
\end{enumerate}

Theorem~\ref{thm:discrepancykomlos} now follows by combining (\ref{komloseqn1}),(\ref{komloseqn2}),(\ref{komloseqn3}) as before and using Observation~\ref{obs:large}. 

\medskip

Theorem~\ref{komlos} now follows easily, by observing that $m$ can be assumed to be polynomially bounded in $n$ and applying a union bound. Indeed, we can discard all rows of $\ell_1$-norm less than $\sqrt{\log n}$ since they can only ever have discrepancy at most $\sqrt{\log n}$. The remaining rows have squared $\ell_2$-norm at least $\frac{\log n}{n}$, as by Cauchy-Schwarz inequality
\[ \sqrt{\log n} \le \sum_{i\in[n]}|b_{ji}| \le \left(\sum_{i\in[n]}b_{ji}^2\right)^{1/2}(n)^{1/2}.
\]
As $\sum_{i,j}b_{ji}^2 \le n$, there can be at most $n^2/\log n$ such rows. We now set $\lambda=O(\sqrt{\log n})$ in Theorem~\ref{thm:discrepancykomlos} and take a union bound over all these rows.

\section*{Acknowledgements}

We would like to thank Nick Harvey, Shachar Lovett, Aleksandar Nikolov, Thomas Rothvoss, Aravind Srinivasan and Mohit Singh for various useful discussions about discrepancy over the years. A special thanks to Aleksandar Nikolov and Kunal Talwar for organizing a recent workshop on discrepancy at American Institute of Mathematics which led to several fruitful discussions related to this work. This work was also done in part while the first author was visiting the Simons Institute for the Theory of Computing.

\bibliographystyle{alpha}
\bibliography{refr}

\end{document}